\documentclass{article}
\usepackage[a4paper]{geometry}
\usepackage[utf8]{inputenc}
\usepackage{amsfonts} 
\usepackage{amsmath}
\usepackage{systeme}
\usepackage{setspace}
\usepackage{cite}
\usepackage{booktabs}
\usepackage{setspace}
\usepackage{array}
\usepackage[utf8]{inputenc}
\usepackage{fancyhdr}
\usepackage{graphicx}
\usepackage{mathtools}
\usepackage{pgfplots}
\pgfplotsset{compat=1.15}
\usepackage{mathrsfs}

\usepackage{pgf,tikz}
\usetikzlibrary{arrows}
\usepackage{cases}
\usepackage{amssymb}
\usepackage{amsthm}
\usepackage{tabularx}
\usepackage{chessboard}
\usepackage{tikz-network}
\usepackage{tikz}
\usepackage{bbm}

\usepackage{mathtools}

\DeclareMathOperator*{\argmin}{arg\,min}

\usetikzlibrary{shapes,backgrounds}
\usetikzlibrary{decorations.markings}
\usepackage{verbatim}
\makeatother
\usepackage[inline]{asymptote}
\usepackage{extarrows}
\usepackage{hyperref}
\hypersetup{
    colorlinks,
    citecolor=black,
    filecolor=black,
    linkcolor=black,
    urlcolor=black
}

\newcolumntype{C}[1]{>{\centering\arraybackslash}m{#1}}
\newtheorem{definition}{Definition}

\newtheorem{theorem}{Theorem}
\newtheorem{corollary}{Corollary}

\newtheorem{assumption}{Assumption}

\title{Pareto Efficient Insurance with Multiple Policyholders, Multiple Insurers, and Multiple Indemnity Environments}
\author{Zijun Meng}
\date{}

\begin{document}
\linespread{1.5}
\setlength{\baselineskip}{18pt}
\maketitle

\begin{center}
    \textbf{Abstract}
\end{center}
This paper proves a sum-minimization characterization of Pareto efficient insurance with multiple policyholders, multiple insurers, and multiple indemnity environments. We also provide a result regarding the pairwise implementability of the policyholder- and insurer-aggregate level arrangements in the multiple policyholders and multiple insurers setting.

\tableofcontents

\newpage
\section{Introduction}
Asimit-Boonen \cite{a2018} proves a sum-minimization characterization of Pareto optimal insurance contracts for agents with translation invariant risk measures under a multiple insurers setting. Asimit-Boonen-Chi-Chong \cite{a2021} studies the problem under a multiple exogenous environments setting. Boonen-Chong-Ghossoub \cite{b2024} proves the characterization with multiple policyholders and a centralized insurer. In each of these papers, only one dimension is multiple and the other two dimensions are degenerated. This paper aims to combine the three directions.

In Section 2, we argue that the sum-minimization characterization is robust when we combine multiple insurers and multiple environments. In fact, the characterization for multiple policyholders and multiple environments follows using the same idea. In Section 3, we study the case when there are simultaneously multiple policyholders and multiple insurers. Instead of specifying the indemnity functions for each pair of policyholder and insurer, we do the characterization in an aggregate level, i.e. how much one has to pay or receive and how much risk one has to bear. Then, we argue in a corollary that this is more general than the previous case and it in particular contains the multiple indemnity environments setting as a special case. We also provide a necessary and sufficient condition for the existence of sets of pairwise contracts that support the aggregate level arrangements in the end of the section.

\section{Multiple Insurers and Multiple Environments}
Consider a probability space $(\Omega,\,\mathcal F,\,\mathbb P)$, denote by $L^p:=L^p(\Omega,\,\mathcal F,\,\mathbb P)$ the set of all $p$-integrable random variables in that probability space and by $L_+^p:=L_+^p(\Omega,\,\mathcal F,\,\mathbb P)\subseteq L^p$ the set of all nonnegative $p$-integrable random variables in that probability space.

Suppose that there is a single policyholder $PH$ with risk measure $\rho^{PH}:\,L^p\to\mathbb R$ and a set $\mathcal R=\{1,\,\dots,\,r\}$ of insurers with risk measures $\rho_1,\,\dots,\,\rho_r:\,L^p\to\mathbb R$ respectively. Suppose that the policyholder is endowed with a risk modelled by a random variable $X\in L_+^p$

\begin{definition}
We say that a risk measure $\rho:\,L^p\to\mathbb R$ is
\begin{itemize}
    \item increasing if $\rho(X)\le\rho(Y)$ whenever $X\le Y$ $\mathbb P$-a.s;
    \item translation invariant if $\rho(X+a)=\rho(X)+a$ for any $a\in\mathbb R$;
    \item law invariant if $\rho(X)=\rho(Y)$ whenever $X\stackrel d=Y$.
\end{itemize}
\end{definition}

\begin{assumption}
$\rho^{PH},\,\rho_1,\,\dots,\,\rho_r$ are translation invariant, law invariant, increasing, and fixs $0$.
\end{assumption}

Let $\Omega=\bigsqcup\limits_{k=0}^m\Omega_k$, indicating a realized riskless environment $\Omega_0$ and $m$ realized risky environments, and $Y=\sum\limits_{k=0}^mk\mathbbm 1_{\Omega_k}$ be an indicator.

Let $S\subseteq\mathcal R$ be the set of all insurers the policyholder trades with, and \[\mathcal I_S:=\{I=(I_k^i)_{i\in S,\,1\le k\le m};\,0\le I^i_k\le id,\,I^i_k\text{ and }R^i_k:=id-\sum_{i\in S}I^i_k\text{ are increasing}\}\] be the set of feasible indemnity tuples. Suppose that the policyholder pays premium $\pi_i$ to and receives a bonus $b_i$ from insurer $i$, so insurer $i$ bears a risk \[X_i:=b_i\mathbbm 1_{\Omega_0}+\sum_{k=1}^mI^i_k(X)\mathbbm 1_{\Omega_k}\] for the policyholder, who has to pay \[\pi:=\sum_{i\in S}\pi_i\] and bear the remaining risk \begin{align*}X^R:&=X-\sum_{i\in S}X_i\\&=X-\sum_{i\in S}\left(b_i\mathbbm 1_{\Omega_0}+\sum_{k=1}^mI^i_k(X)\mathbbm 1_{\Omega_k}\right)\\&=X-\left(\sum_{i\in S}b_i\right)\mathbbm 1_{\Omega_0}-\sum_{k=1}^m\left(\sum_{i\in S}I^i_k(X)\right)\mathbbm 1_{\Omega_k}\\&=-\left(\sum_{i\in S}b_i\right)\mathbbm 1_{\Omega_0}+\sum_{k=1}^m\left(X-\sum_{i\in S}I^i_k(X)\right)\mathbbm 1_{\Omega_k}\\&=-\left(\sum_{i\in S}b_i\right)\mathbbm 1_{\Omega_0}+\sum_{k=1}^mR^i_k(X)\mathbbm 1_{\Omega_k}.\end{align*} Let \[\mathcal C_{S,\,X}:=\{(b_i,\,\pi_i,\,X_i)_{i\in S}\in\mathbb R_+\times\mathbb R_+^S\times L_+^p;\,\rho_i(X_i-\pi_i)\le 0\,\forall\,i\in S,\,\rho^{PH}(X^R+\pi)\le\rho^{PH}(X)\}\] be the set of all incentive compatible tuple of contracts, \begin{align*}\mathcal P_{S,\,X}:=\{&(b_i,\,\pi_i,\,X_i)_{i\in S}\in\mathcal C_{S,\,X};\,\text{there does not exists } (\tilde\pi_i,\,\tilde X_i)\text{ such that}\\&\begin{cases}\rho_i(\tilde X_i-\tilde\pi_i)\le\rho_i(X_i-\pi_i)\forall\,i\in S\\\rho^{PH}(\tilde X^R+\tilde\pi)\le\rho^{PH}(X^R+\pi)\end{cases}\\&\text{with at least one strict inequality}\}\end{align*} be the set of all Pareto optimal tuple of contracts, \[\mathcal S_{S,\,X}:=\argmin_{(b_i,\,\pi_i,\,X_i)_{i\in S}\in\mathcal C_{S,\,X}}\left(\sum_{i\in S}\rho_i(X_i)+\rho^{PH}(X^R)\right)\] be the set of all tuple of contracts minimizing the sum of risk measures of insurers and the policyholder. We have the following result:

\begin{theorem}
We have \[\mathcal S_{S,\,X}=\mathcal P_{S,\,X}.\]
\end{theorem}

\begin{proof}
We first prove that \[\mathcal S_{S,\,X}\subseteq\mathcal P_{S,\,X}.\] Indeed, suppose on the contrary that there exists \[(b_i,\,\pi_i,\,X_i)_{i\in S}\in\mathcal S_{S,\,X}\setminus\mathcal P_{S,\,X}.\] Then, there exists $(\tilde b_i,\,\tilde\pi_i,\,\tilde X_i)_{i\in S}$ so that \[\begin{cases}\rho_i(\tilde X_i-\tilde\pi_i)\le\rho_i(X_i-\pi_i)\forall\,i\in S\\\rho^{PH}(\tilde X^R+\tilde\pi)\le\rho^{PH}(X^R+\pi)\end{cases}\] with at least one strict inequality. Summing these $|S|+1$ inequalities yields \begin{align*}\sum_{i\in S}\rho_i(\tilde X_i)+\rho^{PH}(\tilde X^R)&=\sum_{i\in S}\rho_i(\tilde X_i-\tilde\pi_i)+\rho^{PH}(\tilde X^R+\tilde\pi)\\&<\sum_{i\in S}\rho_i(X_i-\pi_i)+\rho^{PH}(X^R+\pi)=\sum_{i\in S}\rho_i(X_i)+\rho^{PH}(X^R),\end{align*} which contradicts the assumption that \[(b_i,\,\pi_i,\,X_i)_{i\in S}\in\mathcal S_{S,\,X}.\]
We next prove that \[\mathcal S_{S,\,X}\supseteq\mathcal P_{S,\,X}.\] Indeed, suppose on the contrary that there exists \[(b_i,\,\pi_i,\,X_i)_{i\in S}\in\mathcal P_{S,\,X}\setminus\mathcal S_{S,\,X}.\] Then, there exists $(\tilde b_i,\,\tilde\pi_i,\,\tilde X_i)_{i\in S}$ so that \[\sum_{i\in S}\rho_i(\tilde X_i)+\rho^{PH}(\tilde X^R)<\sum_{i\in S}\rho_i(X_i)+\rho^{PH}(X^R).\] We take \[\hat\pi_i=\tilde\pi_i+\rho_i(\tilde X_i-\pi_i)-\rho_i(X_i-\pi_i),\] then \[\rho_i(\tilde X_i-\hat\pi_i)=\rho_i(\tilde X_i)-\tilde\pi_i-\rho_i(\tilde X_i-\tilde\pi_i)+\rho_i(X_i-\pi_i)=\rho_i(X_i-\pi_i)\le 0\] and \begin{align*}\rho^{PH}(\tilde X^R+\hat\pi)&<\sum_{i\in S}\rho_i(X_i)+\rho^{PH}(X^R)-\sum_{i\in S}\rho_i(\tilde X_i)+\sum_{i\in S}(\tilde\pi_i+\rho_i(\tilde X_i-\tilde\pi_i)-\rho_i(X_i-\pi_i))\\&=\sum_{i\in S}\rho_i(X_i)+\rho^{PH}(X^R)-\sum_{i\in S}\rho_i(\tilde X_i)+\sum_{i\in S}\rho_i(\tilde X_i)-\sum_{i\in S}\rho_i(X_i)+\sum_{i\in S}\pi_i\\&=\rho^{PH}(X^R+\pi)\le\rho^{PH}(X),\end{align*} so $(\tilde b_i,\,\hat\pi_i,\,\tilde X_i)_{i\in S}$ is strictly better than $(b_i,\,\pi_i,\,X_i)_{i\in S}$ in the Pareto sense, which contradicts the assumption that \[(b_i,\,\pi_i,\,X_i)_{i\in S}\in\mathcal P_{S,\,X}.\]
\end{proof}

Indeed, following the same proof, we can deduce that the sum-minimization characterization is true for multiple policyholders and multiple environments (with a single insurer).

\section{Multiple Policyholders, Multiple Insurers, and Multiple Indemnity Environments}
Now, suppose that we have a set $\mathcal H=\{1,\,\dots,\,h\}$ of policyholders endowed with risks $X_1,\,\dots,\,X_h\in L_+^p$ and suffer from risk measures $\rho_1^{PH},\,\dots,\,\rho_h^{PH}:\,L^p\to\mathbb R$ and a set $\mathcal R=\{1,\,\dots,\,r\}$ of insurers with risk measures $\rho_1,\,\dots,\,\rho_r:\,L^p\to\mathbb R$ respectively. We impose the same assumption as in the previous section:

\begin{assumption}
$\rho_1^{PH},\,\dots,\,\rho_h^{PH},\,\rho_1,\,\dots,\,\rho_r$ are translation invariant, law invariant, increasing, and fixs $0$.
\end{assumption}

We model the arrangement as a tuple $(d,\,p,\,q)\in(L_+^p)^{h+r}\times\mathbb R^{h+r}$, where $B_i(d)\in L_+^p$ is the retained risk of policyholder $i\in\mathcal H$ after a total payment of $p_i$, $S_j(d)$ is the risk the insurer $j\in\mathcal R$ has to bear with a monetary compensation $q_j$, we assume that $B_i(d)\le X_i$, $p\ge 0$, and $q\ge 0$.

Let \begin{align*}\mathcal C:=\{(d,\,p,\,q)\in(L_+^p)^{h+r}\times\mathbb R^{h+r};\,&\rho_i^{PH}(B_i(d)+p_i)\le\rho_i^{PH}(X_i)\,\forall\,i\in\mathcal H;\\&\rho_j(S_j(d)-q_j)\le0\,\forall\,j\in\mathcal R;\,\sum_{i\in\mathcal H}p_i=\sum_{j\in\mathcal R}q_j\}\end{align*} be the set of incentive compatable and cash clearing arrangements, \begin{align*}\mathcal P:=\{&(d,\,p,\,q)\in\mathcal C;\,\text{there does not exists }(\tilde d,\,\tilde p,\,\tilde q)\in\mathcal C\text{ such that}\\&\begin{cases}\rho_i^{PH}(B_i(\tilde d)+\tilde p_i)\le\rho_i^{PH}(B_i(d)+p_i)\,\forall\,i\in\mathcal H\\\rho_j(S_j(\tilde d)-\tilde q_j)\le\rho_j(S_j(d)-q_j)\,\forall\,j\in\mathcal R\end{cases}\\&\text{with at least one strict inequality}\}\end{align*} be the set of Pareto optimal arrangements, \[\mathcal S:=\argmin_{(d,\,p,\,q)\in\mathcal C}\left(\sum_{i\in\mathcal H}\rho_i^{PH}(B_i(d))+\sum_{j\in\mathcal R}\rho_j(S_j(d))\right)\] be the set of all arrangements minimizing the sum of risk measures of policyholders and insurers. We have the following result:

\begin{theorem}
We have \[\mathcal S=\mathcal P.\]
\end{theorem}

\begin{proof}
We first prove that \[\mathcal S\subseteq\mathcal P.\] Indeed, suppose on the contrary that there exists $(d,\,p,\,q)\in\mathcal S\setminus\mathcal P$. Then, there exists $(\tilde d,\,\tilde p,\,\tilde q)$ so that \[\begin{cases}\rho_i^{PH}(B_i(\tilde d)+\tilde p_i)\le\rho_i^{PH}(B_i(d)+p_i)\,\forall\,i\in\mathcal H\\\rho_j(S_j(\tilde d)-\tilde q_j)\le\rho_j(S_j(d)-q_j)\,\forall\,j\in\mathcal R\end{cases}\] with at least one strict inequality. Summing these $h+r$ inequalities yields \begin{align*}\sum_{i\in\mathcal H}\rho_i^{PH}(B_i(\tilde d))+\sum_{j\in\mathcal R}\rho_j(S_j(\tilde d))&=\sum_{i\in\mathcal H}\rho_i^{PH}(B_i(\tilde d)+\tilde p_i)+\sum_{j\in\mathcal R}\rho_j(S_j(\tilde d)-\tilde q_j)\\&<\sum_{i\in\mathcal H}\rho_i^{PH}(B_i(d)+p_i)+\sum_{j\in\mathcal R}\rho_j(S_j(d)-q_j)\\&=\sum_{i\in\mathcal H}\rho_i^{PH}(B_i(d))+\sum_{j\in\mathcal R}\rho_j(S_j(d)),\end{align*} which contradicts the assumption that $(d,\,p,\,q)\in\mathcal S$. We next prove that \[\mathcal S\supseteq\mathcal P.\] Indeed, suppose on the contrary that there exists $(d,\,p,\,q)\in\mathcal P\setminus\mathcal S$. Then, there exists $(\tilde d,\,\tilde p,\,\tilde q)$ so that \[\sum_{i\in\mathcal H}\rho_i^{PH}(B_i(\tilde d))+\sum_{j\in\mathcal R}\rho_j(S_j(\tilde d))<\sum_{i\in\mathcal H}\rho_i^{PH}(B_i(d))+\sum_{j\in\mathcal R}\rho_j(S_j(d)).\] Let \[\Delta:=\left(\sum_{i\in\mathcal H}\rho_i^{PH}(B_i(d))+\sum_{j\in\mathcal R}\rho_j(S_j(d))\right)-\left(\sum_{i\in\mathcal H}\rho_i^{PH}(B_i(\tilde d))+\sum_{j\in\mathcal R}\rho_j(S_j(\tilde d))\right)>0.\] Take \[\hat p_i:=p_i+\rho_i^{PH}(B_i(d))-\rho_i^{PH}(B_i(\tilde d))-\frac{\Delta}{h+r}\] and \[\hat q_j:=q_j+\rho_j(S_j(\tilde d))-\rho_j(S_j(d))+\frac{\Delta}{h+r}.\] Then \[\rho_i^{PH}(B_i(\tilde d)+\hat p_i)=p_i+\rho_i^{PH}(B_i(d))-\frac{\Delta}{h+r}<\rho_i^{PH}(B_i(d)+p_i)\le\rho_i^{PH}(X_i)\] and \[\rho_j(S_j(\tilde d)-\hat q_j)=\rho_j(S_j(d))-q_j-\frac{\Delta}{h+r}<\rho_j(S_j(d)-q_j)\le 0,\] so $(\tilde d,\,\hat p,\,\hat q)$ is strictly better than $(d,\,p,\,q)$ in the Pareto sense, which contradicts the assumption that $(d,\,p,\,q)\in\mathcal P$.
\end{proof}

We immediately have the following:

\begin{corollary}
Let $\Omega=\bigsqcup\limits_{k=0}^m\Omega_k$, indicating a realized riskless environment $\Omega_0$ and $m$ realized risky environments, and $Y=\sum\limits_{k=0}^mk\mathbbm 1_{\Omega_k}$ be an indicator. Let $I_{ijk}:\mathbb R_+\to\mathbb R_+$ be the indemnity function for policyholder $i\in\mathcal H$, insurer $j\in\mathcal R$ and environment $1\le k\le m$. Let $b_{ij}\ge 0$ be the bonus for policyholder $i\in\mathcal H$, insurer $j\in\mathcal R$ and environment $0$. If $0\le I_{ijk}\le id$, $I_{ijk}$ and $R_{ijk}:=id-I_{ijk}$ are increasing. Then the set of Pareto optimal arrangements is equal to the set of arrangements that minimize the sum of the risk measures of the policyholders and insurers.
\end{corollary}

\begin{proof}
Take \[B_i(d)=X_i-\mathbbm 1_{\Omega_0}\sum_{j\in\mathcal R}b_{ij}-\sum_{k=1}^m\mathbbm 1_{\Omega_k}\sum_{j\in\mathcal R}I_{ijk}(X_i)\] and \[S_j(d)=\mathbbm 1_{\Omega_0}\sum_{j\in\mathcal R}b_{ij}+\sum_{k=1}^m\mathbbm 1_{\Omega_k}\sum_{j\in\mathcal R}I_{ijk}(X_i).\] The result follows from the previous theorem.
\end{proof}

The above theorem analyzes the situation in a policyholder-aggregate and insurer-aggregate level, but it ignores each specific pair $(i,\,j)\in\mathcal H\times\mathcal R$. A natural question is whether an arrangement $(d,\,p,\,q)$ is \textsl{bilaterally implementable}, i.e. there exists a matrix $\pi=(\pi_{ij})_{(i,\,j)\in\mathcal H\times\mathcal R}$ of payments, where $\pi_{ij}\ge 0$ is the payment from policyholder $i$ to insurer $j$ such that the cash clears, i.e. \[\sum_{i\in\mathcal H}\pi_{ij}=q_j\,\forall\,j\in\mathcal R\text{ and }\sum_{j\in\mathcal R}\pi_{ij}=p_i\,\forall\,i\in\mathcal H,\] and everyone is weakly better off, i.e. \[\sum_{j\in\mathcal R}\pi_{ij}\le\rho_i^{PH}(X_i)-\rho_i^{PH}(B_i(d))\] for any policyholder $i\in\mathcal H$ and \[\sum_{i\in\mathcal H}\pi_{ij}\ge\rho_j(S_j(d))\] for any insurer $j\in\mathcal R$. In fact, by the max-flow min-cut theorem, we have the following result:

\begin{theorem}
An arrangement $(d,\,p,\,q)\in\mathcal C$ is bilaterally implementable if and only if for any $B\subseteq\mathcal R$, we have \[\sum_{j\in B}\rho_j(S_j(d))\le\sum_{i\in\mathcal H:\,\exists j\in B:\,\pi_{ij}>0}(\rho_i(X_i)-\rho_i(B_i(d))).\]
\end{theorem}

\end{document}